\documentclass{article}
\usepackage{float}
\usepackage{authblk}

\title{Quantum Money Generated by Multiple Untrustworthy Banks}
\author[1]{Yuichi Sano\thanks{\texttt{sano.yuichi.77v@st.kyoto-u.ac.jp}}} 
\affil[1]{Department of Nuclear Engineering, Kyoto University, Nishikyo-ku, Kyoto 615-8540, Japan}
\date{\today}

\usepackage{graphicx}
\usepackage{amsmath,amsthm,amssymb}
\usepackage{amsfonts}
\usepackage{siunitx}
\usepackage{listliketab}
\usepackage{qcircuit}
\usepackage{braket}
\usepackage{cite}
\usepackage{comment}
\usepackage{physics}

\theoremstyle{definition}
\newtheorem{thm}{Theorem}
\newtheorem{defn}[thm]{Definition}
\newtheorem{lemma}[thm]{Lemma}
\newtheorem{ass}[thm]{Assumption}

\begin{document}

\maketitle

\begin{abstract}
While classical money can be copied, it is impossible to copy quantum money in principle, with only the bank that issues it knowing how to generate it, meaning only the bank can make exact copies.
Not all reliable banks, such as central banks, will issue quantum money, so there is the possibility that untrustworthy banks are distributing fake or multiple copies of the same quantum money without the users' knowledge. 
As such, we propose a {\it quantum patchwork money} scheme in which banks cannot distribute exact copies to users. 
This scheme involves multiple banks providing public-key quantum money as shards and generating quantum patchwork money by combining them.
The banks can use the quantum patchwork money without completely trusting the other banks.
In addition, nonbank users can use safely the quantum patchwork money without trusting any banks potentially focused on self-interest by adding a protocol for monitoring the distribution of copies.
\end{abstract}

\section{Introduction}
Many new applications that involve combining quantum physics and computer science have recently been proposed.
In terms of algorithms, Shore and Grover proposed algorithms for quantum computers that are faster than those for classical computers\cite{Shor,Grover}.
Blind computation is another application unique to quantum computation---one that allows users to delegate calculations for the cloud server without telling the server what the users want to calculate\cite{BFK,Mahadev}.
In the field of cryptography, quantum key distribution is known for its information-theoretic security, which cannot be achieved via cryptography using classical communication\cite{BB84,E91}, while quantum money that is copy-resistant in principle is another important application\cite{Wiesner,BBBW,Aaronson,Farhi,Lutomirski,AaC,QuantumLightning,Qbitco,One-Shot Signatures}.
In this paper, rather than propose a scheme for generating quantum money, we propose an unprecedented application using quantum money that can be used without trusting banks.

\subsection{Classical Money and Quantum Money}
Various tokens have been used in place of money in the past.
Since money can be exchanged for other goods, it must have certain properties that cannot be easily copied.
In the past, precious metals that were difficult to duplicate were chosen as the material for money, while at present, various physical tokens and electronic information are often used as instead.
Examples of physical tokens include the paper money and coins issued by central banks, which are designed to be copy-resistant; however, it is almost impossible to completely prevent copying.
Meanwhile, examples of tokens that use electronic information are credit cards and bank accounts, which again are not completely resistant to copying. 
Therefore, the banks and credit companies that mediate electronic information must be trustworthy institutions.

Quantum money, which, unlike classical money, is copy-resistant in principle, was first proposed by Wiesner\cite{Wiesner}.
The no-cloning theorem in quantum mechanics states that arbitrary quantum states cannot be copied, and Wiesner's quantum money has been proven to have a copy success probability of $(3/4)^n$\cite{wies-sec}.
However, this quantum money has to be checked each time by the bank that issued it to ensure that the token is correct.
A number of so-called ``public-key'' quantum money schemes have been proposed, which allow the users to verify that the token is correct without having to contact the bank\cite{Aaronson,Farhi,Lutomirski,AaC,QuantumLightning,Qbitco,One-Shot Signatures}.
While public-key quantum money is largely copy-resistant, the cryptographic part is based on computational difficulty assumptions.
Quantum money that cannot be copied is a vital quantum application.

While the above are examples of monetary forms that are issued by banks, Bitcoin is decentralized money involving the use of the blockchain technology proposed by Nakamoto\cite{Nakamoto} that are not issued by banks and are not intrinsically copy-resistant since they are essentially tokens of electronic information.
In addition, Bitcoin transactions are not brokered by a trusted bank.
Instead, a distributed ledger based on the blockchain can be used to trace the movement of all money.
By having a third party known as a ``miner'' to monitor the transactions, the use of multiple copies for trading can be prevented. 
Therefore, the Bitcoin scheme does not require the trust of banks or intermediaries, but it does involve huge computational resources.

\subsection{Aims and Objectives}
The aim of this research is to propose a money scheme that can be used without trusting the banks that issue the money, which generally occupy a privileged position.
If the user does not have the ability to verify the money, they will not be able to tell when the bank gives them counterfeit money.
Banks are also free to make exact copies of the money, which applies to both classical and quantum money.
While quantum money cannot be effectively copied by an outsider, the bank knows how to generate it and can make copies at any time.

Most institutions that issue money, such as central banks, are trustworthy.
However, there are other items besides legal tender that people expect to have the same properties as money, such as stocks, bonds, art, and trading card game's(TCG) cards, which require {\it completeness}, {\it soundness}, and {\it copy-resistance} and, yet, much unlike legal tender, are not always dealt with by trustworthy institutions.

What would happen if users used a token issued by an untrustworthy bank?
For example, the bank may issue tokens with the same index as a token held by someone the bank does not like, essentially making them worthless. 
Even if the bank promises to issue a limited number of tokens, it may issue and sell more tokens than promised.
Furthermore, the bank may issue a fake token and pass it on to users as a correct token.
Given that these issues cannot be effectively addressed using the current schemes, our aim is to develop a scheme wherein users can use tokens safely without trusting banks.

\subsection{Contributions}
The two main contributions of this study are as follows:
\begin{itemize}
\item A {\it quantum patchwork money} scheme is proposed in which several untrustworthy banks provide shards to each other to issuein view of issuing quantum money, and. Here the banks that participate in the issuance can use the quantum patchwork money securely without trusting the other banks.
\item The quantum patchwork money scheme is extended to allow nonbank users to use quantum money securely without trusting the banks under the condition of multiple banks issuing quantum money while prioritizing their own interests.
\end{itemize}

\subsubsection{Security}
The following two security criteria are defined for the safe use of money issued by untrustworthy the banks.
\begin{defn}[Informal]
Money that meets the following conditions is referred to as money that can be used without trusting banks.
\begin{enumerate}
   \item The banks cannot offer fake money to any users.
   \item The banks cannot circulate copies of money without the users' permission.
\end{enumerate}
\end{defn}

Here, the first condition is important since untrustworthy banks can potentially offer users fake tokens. 
For example, if a user is given a fake painting by a bank, the user will not be able to reject it if they do not have the ability to judge its authenticity. 
If the user then puts the fake painting up for sale, no one will buy it, and the user will lose money.
In the past, most classical tokens could not be easily verified for correctness; special skills are required for the verification.
However, with some of the current tokens, such as public-key quantum money, users can easily verify their authenticity.

The second condition is crucial to preventing exact copies from being distributed beyond the rules.
For example, if there are two or more bills with the same serial number in circulation, all but one of them will be regarded as counterfeit.
In other cases, the bank may promise to issue a limited number of specific tokens(e.g., TCG cards) but then break its promise and issue more cards. 
Classical tokens complicate the situation since copies of the tokens can be made not only by the bank but also by outsiders.
However, in the case of quantum money, only the bank can generate multiple identical tokens, meaning the focus is on whether the bank is prepared to commit fraud.

A scheme based on public-key quantum money could be beneficial to meeting the above two conditions.
Public-key quantum money can be easily verified as correct by users, and no copies can be made except by the bank.
Therefore, to achieve full security, a scheme is needed to prevent the circulation of copies generated by the banks to anyone other than the users who have consented to this activity.
As such, our scheme is structured in two stages, comprising banks that issue money and users who do not. 
The term ``bank'' is used to define the issuer of money, while the issuer can also be regarded as a ``user'' of quantum money.

\subsubsection{Our Scheme}
In the first scheme, users are limited to the banks.
To start, each bank works with the other banks to generate quantum patchwork money.
At this point, each bank sends its own public-key quantum money to the other banks as a ``shard'' of quantum money before publishing the key corresponding to the shard sent. 
The bank that has received the shards then uses its public key to verify that it is correct. If the bank successfully verifies the shard, it will add additional shards to complete the quantum patchwork money; if it fails to verify the shard, it will ask the bank that generated the failed shard to send a new one.
If the shard demonstrates sufficient completeness, soundness, and copy- resistance, then quantum money completed via this protocol will also have these properties, meaning each shard can only be generated in the bank. 
When the user receives a quantum patchwork money token, they will only accept it if all the included shards can be accepted.
If a bank wishes to copy this completed quantum money token, it will only copy the shard the bank generated, not the whole token. 
Therefore, if a bank wants to make copies of quantum money, it must cooperate with the other banks that have created the shards.
As such, all banks will have to cooperate to make perfect copies, but that will not be a problem since the copies will only be distributed among the users who agree to make copies.
Thus, banks that participate in generating quantum money can use it without trusting the other banks. 
The extended scheme is based on the above scheme but is slightly more complicated since the {\it reissue} protocol is considered.

To extend the scheme, a reissue protocol is introduced. 
Any token can be damaged following frequent use, and provided it is legal tender, the user can take it to a bank to have it exchanged. 
If the token is essentially electronic information, it will be repaired via error correction.
Quantum money is also highly susceptible to errors since it is composed of various quantum states, meaning quantum error correction is a necessary aspect\cite{fault1,fault2,fault3,fault4}.
However, there is the possibility that errors may occur that cannot be addressed via error correction since this method is not perfect.
There is also the possibility of losing one of the shards to another user in the exchange of quantum money.
At such times, the user uses a so-called ``reissue'' protocol to restore the quantum money token.
Here, the user must first indicate to the bank that issued the lost shard that all other shards are correct via quantum communication.
If the bank confirms that all of the other shards are correct, it will issue a new shard with a new public key rather than the original shard.
Because if the bank reissues the original shard, the user can generate multiple identical quantum money by reissuing all the shards.
All users will be aware of the reissue as the public key will also be updated.

It is also assumed that an extended scheme will also involve users who are not banks. 
Since users who are not banks have no privileges regarding quantum money, they will be at a significant disadvantage if they do not place any restrictions on the banks.
\begin{defn}[Informal]
Money that meets the following conditions is referred to as money that can be used without trusting banks that put their own interests first.
\begin{enumerate}
   \item The banks cannot offer fake money to any users.
   \item The banks do not circulate copies of the tokens without the users' permission when the former do not profit from the circulation of quantum money.
\end{enumerate}
\end{defn}
On adopting the second type of security outlined above, we can propose a scheme that allows users to monitor the behavior of the banks and use quantum money without trusting the banks.
The details of the scheme are explained under the assumption that the bank will prioritize its own interests.
The protocol for issuing quantum money is the same as that in the first scheme, as are the banks' security aspects. 
Users need to monitor whether the banks are issuing and circulating copies of quantum money without permission and for it can use the reissue protocol.
This involves updating the public key associated with each piece of quantum money, meaning users possessing copies with updated public keys will become aware that copies are in circulation. 
In other words, users can discover the circulation of copies by performing the reissue protocol with a reasonable probability.
However, this protocol alone cannot prevent banks from issuing copies.
The second scheme is where it becomes essential to demand a fee from the user for reissue and a reserve of quantum money from the bank.
Since banks are limited in their capacity to distribute copies, users can prevent bank fraud by ensuring that the banks' profit gained from reissuing is greater than that gained from fraudulent activity and that the loss incurred through the reserve is larger than the profit gained via the fraud.

\subsection{Comparison With Bitcoin}
While we propose a unique money scheme that users can adopt without trusting their banks, a similar money scheme already exists in the form of Bitcoin\cite{Nakamoto}.
The main difference between our quantum patchwork money scheme and Bitcoin relates to whether banks issue money, while the most significant similarity is that both schemes are largely resistant to copying.
Given that bitcoins are essentially digital information, they can be copied; however, by recording the distribution of all tokens as a blockchain, users can prevent the distribution of copied tokens.

Bitcoin has several disadvantages for its security.
The first is that users must pay fees to miners who monitor the circulation of bitcoins.
Meanwhile, in our scheme, while the user may have to pay a reissue fee for security reasons, the reissue can be free in exceptional cases.
However, after 2140, Bitcoin will only reward miners with fees, so users will always need to pay a fee to trade using bitcoins.
The second disadvantage is that Bitcoin has a limited transaction rate\cite{transaction1,transaction2}, i.e., the transaction must be approved by the miners.
Bitcoin transactions are incorporated into the blockchain by a miner winning a game known as ``proof of work''\cite{transaction2,tran-time}.
Users are encouraged to wait for the game to finish around six times to prevent fraud, meaning around one hour is required to complete the transaction.
In the quantum patchwork money scheme, the money transactions are between the users, meaning there is no limit to the transaction rate other than the physical limit.
In addition, the users know immediately if the transaction will be successful since they only need to use the public- key to verify that the quantum money token is correct.
The third disadvantage is that Bitcoin mining requires high electricity and computing power.
In fact, the amount of electricity consumed in terms of Bitcoin mining is an issue, estimated to have accounted for around $0.5\% -- 1\%$ of the world's electricity consumption in 2021\cite{elec1,elec2}.
The miners are increasing their computational resources and power consumption every year to succeed in mining.
In contrast, the only energy used in our scheme is that used for making, sending, and protecting quantum money.
A further disadvantage is that all Bitcoin transactions are public, and since they are recorded in the blockchain, it has been suggested that the privacy of the users could be breached through examining the transactions\cite{bitco-ano1,bitco-ano2,bitco-ano3}.
However, since our scheme is a quantum money scheme involving users trading with each other locally, much like with legal tender, there is no possibility of being tracked by a third party.

\subsection{Related Works and Open Problems}
The main goal of our proposed scheme is to prevent the circulation of copies of quantum money; however, there exists prior research on preventing the actual generation of quantum money copies\cite{Qbitco}.
The security of not allowing banks to generate copies of quantum money extends the security related to not allowing copies of quantum money to be circulated.
Prior research has focused on using blockchain technology to prevent only one bank from generating quantum money.
However, its scheme imposes a substantial restriction that the banks do not cooperate.
It is a significant open problem whether a quantum money scheme prohibits banks from making copies without imposing substantial restrictions.
Here, it can be stated that quantum lightning\cite{QuantumLightning}, quantum money that can only be created once, and one-shot  signatures\cite{One-Shot Signatures}, which are essentially one-time signatures, offer essential perspectives to resolving this open problem.

Another related topic relates to quantum copy protection\cite{Aaronson,CP1,CP2,CP3,CP4}, an application first proposed in \cite{Aaronson} to prevent users from copying loaned quantum software.
This application can be thought of as a quantum version of copy protection, one that prevents the piracy of classical software and one that may have a synergy with our quantum patchwork money.
Let us consider a quantum software with quantum copy protection instead of a money value counterpart to quantum patchwork money.
The bank uses quantum patchwork money to guarantee that there is only one quantum software in the world.
The user can be sure that there is only one unique quantum software using the quantum patchwork money scheme.
At the same time, the bank can be sure that the quantum software it lends is the only one of its kind since it is guaranteed to not be copied. 
Here, we consider the notion that this idea can be applied to applications other than quantum copy protection.

\section{Preliminaries}
In this section, we introduce the public-key quantum money scheme that forms part of our scheme.

\subsection{Public-Key Quantum Money}
In this paper, we deal with a public-key quantum money scheme with sufficient properties as a ``black box.'' 
The specific structure of the public-key quantum money scheme has been discussed in previorus research\cite{Aaronson,Farhi,Lutomirski,AaC,QuantumLightning}.

First, we must define the properties of {\it money} that must be satisfied.

\begin{defn}[Money]
\label{def:money}
``Money'' is defined in terms of a token that satisfies the following properties:
\begin{description}
    \item[Completeness] A verifier accepts any valid token with  a probability of $1 - \epsilon$, where $\epsilon$ is less than $1/2$.
    \item[Soundness] The verifier only accepts a random incorrect token dependeintg on only the publicly available information with at most a negligible probability of $\lambda$.
    \item[Copy-resistance] The possibility of generating $q$ ($ > p$) tokens accepted by the verifier from $p$ valid tokens using any device that is dependent on only publicly available information is at most a negligible probability of $\delta$.
\end{description}
Money is regarded as having \textit{\textbf{perfect completeness}} when $\epsilon = 0$.
\end{defn}
It should be noted that while our definition of money distinguishes between soundness and copy- resistance, the two properties can be taken together as soundness since copy- resistance is equivalent to soundness when $p = 0$.
In the following sections, we refer to $\epsilon$ in completeness as completeness error, $\lambda$ in soundness as soundness error, and $\delta$ in copy resistance as copy-resistance error.

Here, the public-key quantum token scheme is described with reference to \cite{AaC}.
\begin{defn}[Public-Key Quantum Token Schemes]
A {\it public-key quantum token scheme} $\mathcal{S}$ consists of three polynomial-time quantum algorithms:
\begin{itemize}
    \item {\bf KeyGen}, which takes as input a security parameter $0^n$, and probabilistically generates a key pair $(k_{private}, k_{public})$.
    \item {\bf Bank}, which takes as input $k_{private}$, and probabilistically generates a quantum state $\sigma$ as token.
    \item {\bf Ver}, which takes as input $k_{public}$ and an alleged token $\phi$, and either accepts or rejects.
\end{itemize}
\end{defn}
If the tokens generated by public-key quantum token scheme $\mathcal{S}$ have the properties of money, $\mathcal{S}$ can be regarded as a {\it public-key quantum money scheme}.
\begin{thm}
Public-key quantum token scheme $\mathcal{S}$ can generate tokens as {\it money} if $\mathcal{S}$ has the following properties:
\begin{itemize}
    \item {\bf Ver} accepts with a probability of $1 - \epsilon$ for all $k_{public}$ and valid tokens, where $\epsilon$ is less than $1/2$.
    \item Let {\bf Count} be a function that receives a collection of alleged tokens $\{\phi_1,\ldots,\phi_r\}$ and counts how many alleged tokens the {\bf Ver} accepts.
    For any quantum circuit $\mathcal{C}(k_{public},\sigma_1,\ldots,\sigma_p)$ of size poly(n), the following inequality holds:
    \begin{equation*}
        \text{Pr}[{\bf Count}(k_{public},\mathcal{C}(k_{public},\sigma_1,\ldots,\sigma_p))>p]\leq\delta,
    \end{equation*}
    where $\delta$ is the negligible probability and the probability is over the key $(k_{private}, k_{public})$, valid tokens $\sigma_1,\ldots,\sigma_p$ generated by {\bf Bank}, and the behavior of {\bf Count} and $\mathcal{C}$.
\end{itemize}
\end{thm}
\begin{proof}
This is easily verified using Definition \ref{def:money}.
\end{proof}

The completeness error is defined to be less than $1/2$ in Definition \ref{def:money}; however, it is known that the completeness error can be exponentially reduced for polynomial parameter $q(n)$.

\begin{lemma}[Reducing Completeness Error\cite{AaC}]
\label{lem:RCE}
Let $\mathcal{S}$ = ({\bf KeyGen},{\bf Bank},{\bf Ver}) be a public-key quantum money scheme with a completeness error of $\epsilon< 1/2$, and a copy-resistance error of $\delta<1-2\epsilon$ when $p=1$. Then for all polynomials $q(n)$ and all $\delta^\prime>\frac{\delta}{1-2\epsilon}$, we can construct amplified public-key quantum money scheme $\mathcal{S}^\prime = ({\bf KeyGen^\prime},{\bf Bank^\prime},{\bf Ver^\prime})$ with a completeness error of $1/2^{q(n)}$ and a copy-resistance error of $\delta^\prime$.
\end{lemma}

It is also known as {\it Almost As Good As New Lemma} that the amplification of completeness error can be slowed down enough even for repeatedly verifying quantum money tokens. 

\begin{lemma}[``Almost As Good As New Lemma''\cite{AaC,Almost}]
\label{lem:Almost As Good}
Let us suppose a measurement on a mixed state $\sigma$ yields a particular outcome with a probability of $1 - \epsilon$.
Then after the measurement, $\tilde{\sigma}$ can be recovered such that $\|\tilde{\sigma}-\sigma\|_{tr} \leq \sqrt{\epsilon}$.
\end{lemma}

With the Almost As Good As New Lemma and Lemma \ref{lem:RCE}, a quantum money token can be made to have an exponentially lower completeness error permanently, which means that the token can be verified after a polynomial number of times. 
If a user can verify the quantum money a polynomial times without failure, we can show the following lemma.

\begin{lemma}
\label{lem:poly}
It is assumed that the quantum money verification for any polynomial number of times does not fail.
The completeness error is less than any polynomial parameter $1/q(n)$ for any confidence level.
\end{lemma}
\begin{proof}
Let $N$ be the number of quantum money verifications.
When $N$ is sufficiently larger than $q(n)$, the distribution can be approximated as a normal distribution\cite{CLT}.
Letting the standard score be $\alpha$, $q(n)$ satisfies 
\[
N \leq \alpha^2(q(n)-1).
\]
Therefore, the number of verifications, N, is polynomial.
\end{proof}

From Lemma \ref{lem:poly}, a user can ensure that the completeness error is less than $1/q(n)$ with a high enough probability by verifying the quantum money the appropriate polynomial number of times.

\section{Quantum Patchwork Money}
In this section, we propose quantum patchwork money schemes generated by multiple untrustworthy banks.
First, we explain the fraud committed by untrustworthy banks, and then we define what properties are required for users to feel safe using the quantum patchwork money generated by these banks.
Next, if users are limited to using banks, we propose a quantum patchwork money scheme that satisfies these properties. 
We also demonstrate that the scheme can be extended to include nonbank users.
Following this, a reissue protocol for repairing any broken quantum patchwork money tokens is defined, demonstrating that the above security is satisfied even when users use the reissue protocol.

\subsection{Quantum Patchwork Money With Only Banks}
We define a security criterion that allows users to securely use quantum money even when the banks cannot be trusted.
\begin{defn}[Tight Security]
When a quantum money scheme satisfies the following conditions, it can be regarded as a quantum money scheme that can be used safely even if the banks cannot be trusted.
\begin{itemize}
\item The money also satisfies the completeness and soundness requirements even when users receive it from the bank.
\item If not all users permit to make copies, banks can only make copies that users who do not permit it accept by at most negligible probability of $\eta$.
\end{itemize}
\end{defn}

The first condition is satisfied for general public-key quantum money, meaning the second condition is essential.
Here, we propose quantum patchwork money that is based on public-key quantum money and satisfies the second condition.
We first define the protocol for generating quantum patchwork money as Protocol 1, and define a verification protocol for quantum patchwork money as Protocol 2.

\begin{table}[tp]
  \centering
  \begin{tabular}{p{12cm}}
    \hline
    {\bf Protocol 1} Generation of Quantum Patchwork Money \\
    \hline
       \begin{minipage}{12cm}
       \vspace{1pt}
       \begin{description}
        \item[\textbf{Step 1.}] \textbf{Choice of bank}\\        
        The bank chooses one bank.
        \item[\textbf{Step 2.}] \textbf{Generation of shards}\\
        Each bank generates a public-key quantum money called a ``shard.''
        Each bank publishes the public-key.
        \item[\textbf{Step 3.}] \textbf{Acceptance of shards}\\
        The bank chosen in step 1 receives a shard from the other banks.
        The bank verifies all shards received.
        If the bank accept all shards, then the next step can be taken; otherwise, one must start over from Step 2.
        \item[\textbf{Step 4.}] \textbf{Generation of quantum patchwork money}\\
        The bank combines the shards it receives with its own shards to form a single quantum patchwork money token.
        All banks sign all public keys by digital signature.
        \item[\textbf{Step 5.}] \textbf{Repetition of generating}\\
        Banks repeat steps 1--4 until all money tokens are generated.
        \\
       \end{description}
    \end{minipage}
    \\
    \hline
  \end{tabular}
\end{table}
\begin{table}[tp]
  \centering
  \begin{tabular}{p{12cm}}
    \hline
    {\bf Protocol 2} Verification of Quantum Patchwork Money \\
    \hline
       \begin{minipage}{12cm}
       \vspace{1pt}
       \begin{description}
        \item[\textbf{Step 1.}] \textbf{Verification of shards}\\
        A sender sends one shard of quantum patchwork money to a verifier, who then verifies the shard.
        If the verifier accepts the shard, they send it back to the sender; otherwise, the verifier terminates this protocol.
        \item[\textbf{Step 2.}] \textbf{Accept of quantum patchwork money}\\
        The verifier repeats Step 1 until all shards are accepted.
        Then, the sender sends the shards to the verifier one by one.
        The verifier then verifies the sent shard and keeps it in hand without sending it back this time.
        If the verifier accepts all shards, they complete the quantum patchwork money; otherwise, the verifier terminates this protocol.\\
       \end{description}
    \end{minipage}
    \\
    \hline
  \end{tabular}
\end{table}

\begin{thm}[Quantum Patchwork Money]
\label{thm:QPM}
Let $m$ be the number of banks that participate in generating a quantum patchwork money scheme.
A quantum patchwork money scheme is the scheme of quantum money generated by means of Protocol 1, where the shard is public-key quantum money with a completeness error of $\epsilon = 1/p(n)$ for $p(n)>2m$, a soundness error that is negligible probability $\lambda$, and a copy-resistance error that is negligible probability $\delta$.
The verification of quantum patchwork money is carried out using Protocol 2. 
This quantum patchwork money scheme is a quantum money scheme.
\end{thm}
\begin{proof}
The soundness and copy-resistance errors in the quantum patchwork money scheme are clearly negligible.
The completeness error of the quantum patchwork money is $\epsilon^\prime = 1 - (1-1/p(n))^m > m/p(n)$, meaning when $p(n)>2m$, then $\epsilon^\prime < 1/2$.
Thus, the quantum patchwork money scheme is a quantum money scheme.
\end{proof}

The completeness error can be exponentially reduced from Lemma \ref{lem:RCE}.
However, the completeness error is not necessarily exponentially low because the bank cannot be trusted.
From Lemma \ref{lem:poly}, the verifier can ensure that the completeness error is $\epsilon < 1/2m$ with a sufficiently high probability, assuming that the verifier never fails to verify any polynomial number of times.
Each bank can prevent other banks from sending shards with completeness errors higher than $1/2m$ by performing a polynomial number of verifications instead of a single verification in step 2 of Protocol 1.
Therefore, it can be assumed that completeness errors are sufficiently low.

The quantum patchwork money scheme is a quantum money scheme that can be used even if the user does not trust the banks.
\begin{thm}[Quantum Patchwork Money Satisfies the Tight Security Conditions]
\label{thm:SQPM}
The quantum patchwork money scheme satisfies the tight security conditions when the users are limited to using banks.
\end{thm}
\begin{proof}
Even if the user receives quantum patchwork money from the bank, the scheme clearly satisfies the completeness and soundness requirements.
Since soundness error $\lambda$ and copy-resistance error $\delta$ of each shard are negligible, even if some banks make copies, the probability that the user will accept a copy is at most a negligible probability of $\eta<\lambda,\delta$.
Of course, if all banks participated in making copies, they could create complete copies, but this would not be a problem since all the users, i.e., the banks, would have given their permission.
Therefore, the quantum patchwork money scheme satisfies the tight security conditions.
\end{proof}

\begin{table}[tp]
  \centering
  \begin{tabular}{p{12cm}}
    \hline
    {\bf Protocol 3} Sign Quantum Patchwork Money\\
    \hline
       \begin{minipage}{12cm}
       \vspace{1pt}
       \begin{description}
        \item[\textbf{1.}] \textbf{The first nonbank user}
        \vspace{-5pt}
        \begin{description}
        \item[Step 1.] Verification of shards\\
        A sender sends one shard of quantum patchwork money to the first nonbank user.
        The user then verifies the shard.
        If the user accepts the shard, they send it back to the sender, otherwise the user terminates this protocol.
        \item[Step 2.] Adding a shard for the signature\\
        The user repeats Step 1 until all shards are accepted. 
        The user generates a shard with the same conditions as the bank's shard as a signature.
        The user sends the signature shard to the sender and distributes the public key.
        If accepted, the sender verifies the signature shard and adds it to the quantum patchwork money; otherwise, the sender repeats until it is accepted.
        \item[Step 3.] Signature to signature\\
        All banks sign the public key for the shard which is the user's signature.
        \item[Step 4.] Repeat Steps 1--3 for all quantum patchwork money  \\
        Steps 1--3 are repeated until the user has signed all quantum patchwork money.
        \vspace{-5pt}
        \end{description}
        \item[\textbf{2.}] \textbf{Other nonbank users}
        \vspace{-5pt}
        \begin{description}
        \item[Step 1.] Verification of shards\\
        A sender sends one shard of quantum patchwork money or its other nonbank user's signature to the nonbank user.
        The user verifies the shard.
        If the user accepts the shard, they send it back to the sender; otherwise, the user terminates this protocol.
        \item[Step 2.] Adding a shard for the signature\\
        The user repeats Steps 1 until all shards are accepted.
        The user generates a shard with the same conditions as the bank's shard as a signature.
        The user sends the signature shard to the sender and distributes the public key.
        If accepted, the sender verifies the signature shard and adds it to the quantum patchwork money; otherwise, the sender repeats until it is accepted.
        \item[Step 3.] Signature to signature\\
        All banks and other nonbank user sign the public key for the shard which is the user's signature.
        \item[Step 4.] Repeat Steps 1--3 for all quantum patchwork money \\
        Steps 1--3 are repeated until the user has signed all quantum patchwork money.\\
        \end{description}
       \end{description}
    \end{minipage}
    \\
    \hline
  \end{tabular}
\end{table}

While it has been demonstrated that each bank could use quantum patchwork money without trusting the other banks, the current scheme is not safe for users other than the banks that participated in generating the quantum patchwork money. 
Therefore, we propose an additional protocol for nonbank users to use without trusting the bank.
In Protocol 3, the nonbank users add their own signature shards to the quantum patchwork money; provided they have the same properties as the bank's shards; this means that after signing, the nonbank users will be in the same position as the bank.
In fact, even if these users do sign for the quantum patchwork money, there is no disadvantage, while there is an advantage in increasing the number of users of quantum patchwork money.

\begin{thm}[Signed Quantum Patchwork Money]
Let $m$ be the number of banks that participated in generating a signed quantum patchwork money scheme.
The signed quantum patchwork money with the signature scheme is the scheme of generating quantum money by means of Protocol 1 and adding a shard for the signature using Protocol 3, where the shard is public-key quantum money with a completeness error of $\epsilon = 1/p(n)$ for $p(n)>2m$, a soundness error that is negligible probability $\lambda$, and a copy-resistance error that is negligible probability $\delta$.
Also, the verification of the signed quantum patchwork money is carried out using Protocol 2.
This signed quantum patchwork money scheme is a quantum money scheme.
\end{thm}
\begin{proof}
This is clear from Theorem \ref{thm:QPM}.
\end{proof}

\begin{thm}[Signed Quantum Patchwork Money Satisfies the Tight Security Conditions]
\label{thm:QPMWS}
The signed quantum patchwork money scheme satisfies the tight security conditions when the users are limited using to banks and the user adds their signature.
\end{thm}
\begin{proof}
This is clear from Theorem \ref{thm:SQPM}.
\end{proof}

Using Protocol 3, even nonbank users can use quantum patchwork money without trusting the bank. 
However, it is extremely time-consuming for every user to sign every quantum patchwork money. 
Therefore, we propose a scheme in the following section that retains the same security for banks but relaxes the security for nonbank users to ensure ease of use.

\subsection{Reissue Protocol}
Up to this point, we had not envisioned the case where the shards that consist of the quantum patchwork money would break.
However, if the completeness error of the shard is nonzero, the shard will be broken with a certain probability of verification.
Therefore, even if only shard is broken, a protocol is needed to repair the quantum patchwork money properly.
The protocol for repairing quantum patchwork money Protocol 4, is labeled a ``reissue protocol.''

\begin{table}[tp]
  \centering
  \begin{tabular}{p{12cm}}
    \hline
    {\bf Protocol 4} Reissue Protocol\\
    \hline
       \begin{minipage}{12cm}
       \vspace{1pt}
       \begin{description}
        \item[\textbf{Step 1.}] \textbf{Verification of other shards}\\
        The user sends one of the remaining shards to the bank that created the broken shard.
        The bank verifies the shard, and if the bank accepts it, the shard is returned to the user; otherwise, the protocol is terminated.
        This is repeated until all shards have been verified.
        \item[\textbf{Step 2.}] \textbf{Generation of a new shard}\\
        If all remaining shards are accepted, the bank generates a new shard and send it to the user.
        The bank then distributes the public-key.
        The user verifies the new shard received from the bank.
        \item[\textbf{Step 3.}] \textbf{Complete reissue}\\
        Step 2 is repeated until the user accepts the new shard.
        Then all banks then sign the public-key.\\
       \end{description}
    \end{minipage}
    \\
    \hline
  \end{tabular}
\end{table}

We show that a quantum patchwork money scheme using the reissue protocol also satisfies security.
\begin{thm}[Quantum Patchwork Money Using the Reissue Protocol Satisfies the Tight Security Conditions]
\label{thm:QPMU}
The quantum patchwork money scheme using the reissue protocol satisfies the tight security conditions.
\end{thm}
\begin{proof}
When $m-2$ users out of $m$ users make copies, this is the same as the proof of Theorem \ref{thm:SQPM}.
Meanwhile, if $m-1$ out of $m$ users make copies, complete copies can be made by using the reissue protocol to change the shards of the banks that did not participate in the fraud to the correct shards.
However, the new shard in Step 2 of the reissue protocol will update its public key.
Thus, the users can make a correct copy, but the original quantum patchwork money will be rejected when verified.
In other words, the position of the copy and the original quantum patchwork money can be regarded as swapped.
Only either the quantum patchwork money or one of its copies will always be accepted.
Thus, quantum patchwork money using the reissue protocol also satisfies the tight security conditions.
\end{proof}

\section{Extended Quantum Patchwork Money}
In the previous section, we showed that a quantum patchwork money scheme could be used by banks and nonbank users without trusting the banks.
However, this type of scheme is inconvenient when many users participate.
Thus, we propose a more accessible scheme using the reissue protocol in exchange for relaxed security.
Specifically, we propose an extended quantum patchwork money scheme that is also safe for nonbank users under the condition that banks prioritize their own interests.

\subsection{Relaxed Security}
In the previous section, we defined the security that is safe for any bank; however, in this subsection, a self-interested supremacy bank, which is considered to be more common, is adopted. 
To start, we must define the concept of a {\it self-interest supremacy bank}.
\begin{defn}[Self-Interested Supremacy Bank]
A self-interested supremacy bank only creates copies of money if the bank itself can make profit $b$ ($b>0$) through such activity.
\end{defn}

In other words, a self-interested supremacy bank will not commit fraud if its action is not beneficial to the bank itself. 
By taking the appropriate actions, nonbank users can reduce the profits and increase the losses that banks incur through fraudulent activity, thus preventing bank fraud. 
Therefore, we define a relaxed security criterion, one that allows users to securely use quantum money even when the self-interested supremacy banks cannot be trusted.
\begin{defn}[Relaxed Security]
When a quantum money scheme satisfies the following conditions, it can be regarded as a quantum money scheme that can be used safely even if the self-interested supremacy banks cannot be trusted.
\begin{itemize}
\item The money satisfies the completeness and soundness requirements even when the users receive it from the bank.
\item The bank's profit $b$ depends on parameters chosen by the users and can be set to $b<0$ if parameters are appropriately chosen.
\end{itemize}
\end{defn}

\subsection{Extended Quantum Patchwork Money}
In this subsection, we propose an extended quantum patchwork money that satisfies the relaxed security conditions.
First, we propose Protocols 5 and 6, which allow nonbank to become banks or banks to become nonbank users.
Here, it should be noted that Protocol 5 is almost the same as Protocol 3, with the difference between the two being the threshold for the amount of quantum patchwork money a user needs to have to gain a position equivalent to that of the banks.
This means that if the threshold is set to zero, Protocol 5 will be the same as Protocol 3. 
The same argument can be made as in Theorem \ref{thm:QPMWS}, since a threshold of zero means that any user can be in the same position as the bank.
In addition, the extended scheme requires that a bank become a nonbank user if the amount of quantum patchwork money it possesses falls below a specific threshold which is where Protocol 6 comes in.
While in practice, a bank does not necessarily have to become a nonbank user, this assumption is adopted to simplify our model. 
Finally, the reissue protocol is changed to Protocol 7, which includes a set cost.

\begin{table}[tp]
  \centering
  \begin{tabular}{p{12cm}}
    \hline
    {\bf Protocol 5} Protocol for Users to Become Banks\\
    \hline
       \begin{minipage}{12cm}
       \vspace{1pt}
        \begin{description}
        \item[Step 1.] Verification of shards\\
        Each bank ensures that the total amount of quantum patchwork money held by a user who wants to become a bank exceeds a specific threshold.
        If the user's quantum money exceeds the threshold, the next step can be taken; otherwise, the protocol is terminated.
        \item[Step 2.] Verification of shards\\
        A bank sends one shard of quantum patchwork money to the nonbank user, who then verifies the shard.
        If the user accepts the shard, the user sends the shard back to the bank; otherwise the user terminates this protocol.
        \item[Step 3.] Adding a shard\\
        The user repeats Step 1 until all shards are accepted.
        The user generates a shard with the same conditions as the bank's shard; sends it to the bank and distributes the public key.
        If accepted, the bank verifies the shard and adds it to the quantum patchwork money; otherwise, the bank repeats the step until it is accepted.
        \item[Step 4.] Signature to signature\\
        All banks sign the public key for the shard which is the user's signature.
        \item[Step 5.] Repeat Steps 2-- for all quantum patchwork money \\
        Steps 2--4 are repeated until the user has added shards to all quantum patchwork money.\\
        \end{description}
    \end{minipage}
    \\
    \hline
  \end{tabular}
\end{table}

\begin{table}[tp]
  \centering
  \begin{tabular}{p{12cm}}
    \hline
    {\bf Protocol 6} Check an Amount of the Quantum Patchwork Money\\
    \hline
       \begin{minipage}{12cm}
       \vspace{1pt}
        \begin{description}
        \item[Step 1.]  Check an amount of the quantum patchwork money\\
        The users check whether or not the total amount of quantum patchwork money held by the bank exceeds a specific threshold.
        If the amount does not exceed the threshold, the next step can be taken; otherwise, the protocol is terminated.
        \item[Step 2.] Erase the shards of the bank\\
        Banks that have an amount below the threshold will be announced to the users, which will then be verified by all users.
        Following the verification, the users will exclude the shard generated by the bank from thier quantum patchwork money.
        \item[Step 2.] Exclude the bank\\
        If the shards are excluded from all quantum patchwork money, the bank becomes a nonbank user.\\
        \end{description}
    \end{minipage}
    \\
    \hline
  \end{tabular}
\end{table}

\begin{table}[tp]
  \centering
  \begin{tabular}{p{12cm}}
    \hline
    {\bf Protocol 7} Reissue Protocol With a Cost\\
    \hline
       \begin{minipage}{12cm}
       \vspace{1pt}
       \begin{description}
        \item[\textbf{Step 1.}] \textbf{Verification of other shards}\\
        The user sends one of the remaining shards to the bank that created the broken shard.
        The bank then verifies the shard, and, if accepted, returns it to the user; otherwise, terminates this protocol is terminated. 
        This step is repeated until all shards have been verified.
        \item[\textbf{Step 2.}] \textbf{Generation of a new shard}\\
        If all the remaining shards are accepted, the bank generates a new shard, sends it to the user, and distributes the public key.
        The user then verifies the new shard received from the bank.
        \item[\textbf{Step 3.}] \textbf{Complete reissue and pay cost}\\
        Step 2 is repeated until the user accepts the new shard, at which point the user pays a fee to the bank.
        Then all banks then sign the public-key.\\
       \end{description}
    \end{minipage}
    \\
    \hline
  \end{tabular}
\end{table}

\begin{thm}[Extended Quantum Patchwork Money]
\label{thm:EQPM}
Let $m$ be the number of banks generating an extended quantum patchwork money scheme that involves generating quantum money by means of Protocol 1, where the shard is public-key quantum money with a completeness error of $\epsilon = 1/p(n)$ for $p(n)>2m$, a soundness error that is negligible probability $\lambda$, and a copy-resistance error that is negligible probability $\delta$.
In addition, the verification of the extended quantum patchwork money is carried out using Protocol 2.
The extended quantum patchwork money scheme includes Protocols 5–-7.
The extended quantum patchwork money scheme is a quantum money scheme.
\end{thm}
\begin{proof}
This is clear from Theorem \ref{thm:QPM}.
\end{proof}

We next show that the extended quantum patchwork money scheme satisfies the
relaxed security conditions, assuming the following two properties for the quantum patchwork money.

\begin{ass}[Assumptions Regarding the Value of the Quantum Patchwork Money]
\label{ass:ass}
The following assumptions regarding the value of the quantum patchwork money are adopted:
\begin{itemize}
\item When $m$ (the number of banks) reaches zero through following Protocol 6, all quantum patchwork money loses its value.
\item When the banks circulates copies without the users' consent and the users become aware of this, the quantum patchwork money loses its value.
\end{itemize}
\end{ass}
These two assumptions are regarded as natural assumptions since quantum patchwork money is essentially money in the first place.
Thus, the security of the extended quantum patchwork money scheme can be confirmed.

\begin{thm}[Extended Quantum Patchwork Money Satisfies the Relaxed Security Conditions]
\label{thm:EQPMRS}
The extended quantum patchwork money scheme satisfies the relaxed security conditions for nonbank users and satisfies the tight security conditions for banks, under Assumption \ref{ass:ass}.
\end{thm}
\begin{proof}
The security between the banks is tight, and the proof is exactly the same as in Theorem \ref{thm:QPMU}.

We then show that the scheme also satisfies the relaxed security when used by nonbank users.
The scheme satisfies the conditions for when a not bank user receives quantum patchwork money from the banks.
We then show that the user can set the bank's profit $b$ to $b<0$ when the bank distributes copies.

To start, we show how users become aware of the distribution of copies.
Here, the users periodically reissue their own quantum patchwork money using Protocol 7.
If copies of the reissued quantum money are in circulation, the owner of the copies will be aware of the fraud since the reissue protocol will update the public key, which is also the case when the owner of the copies reissues them.
This means that the users can monitor the banks for fraud through the reissue protocol.
However, the fact that users can detect bank fraud is not enough to prevent the circulation of bank copies.
Therefore, we next show an incentive for banks to keep copies out of circulation by allowing users to make costly reissues.

Let $t$ be the threshold amount of quantum patchwork money the bank requires in Protocols 5 and 6, and let $V$ be the value of one quantum patchwork money.
Then the total value of the quantum money held by all the banks is $mtV$.
In addition, let $r$ be the bank's revenue obtained by distributing copies.
Then, we estimate the expected revenue that the bank would have earned if it had not committed fraud.
Let $\tau$ be the frequency per day with which a user reissues one of the quantum patchwork money, let $u$ be the total amount of all the users' quantum patchwork money, and let the reissue cost in protocol 7 be $C = cV$ where $V$ is the value  of the quantum patchwork money, and $c$ is a constant coefficient.
At this point, the daily reissue cost paid by the user to the bank is $\tau cVu$.
The expected revenue that the bank will earn from the reissue protocol can be  estimated to be $\alpha\tau cVu$ using a constant coefficient $\alpha$.

We then estimate the profit the bank will make by circulating the copies, while we estimate the number of copies that can be distributed before the users become aware of it.
Let $Y$ be the number of copies the banks can give to the users per day.
Since the probability of a reissue for a random copy on a given day is $1/\tau$, the expected value of the number of days until the circulation of the copies is exposed is $1-(1-1/\tau)^Y$.
If $Y$ is large enough, the circulation of copies will be exposed within one day, meaning the profit the banks can make from copying is $YV$.

Finally, the difference between the profit and loss the banks make on the circulation of the copies can be estimated.
Based on the aforementioned assumption, if users learn that a copy is being distributed, $V=0$, meaning $mtV=0$.
In other words, the sum of the number of assets and revenue the banks lose by circulating copies will be $mtV + \alpha\tau cVu$.
Therefore, the difference between the bank's profit and loss is $b=YV - (mtV + \alpha\tau cVu)$.
The parameters that the users can determine are the frequency of reissue protocol $\tau$ and constant coefficient $c$ of the cost.
When $(Y-mt)/\alpha u<\tau c$, the scheme satisfies the relaxed security conditions.
If $mt>Y$, then the cost for the reissue protocol is zero.
\end{proof}

\section*{Acknowledgment}
We would like to thank Takayuki Miyadera for many helpful comments.
This work was supported by JST SPRING, Grant Number JPMJSP2110.

\end{document}